\documentclass[a4paper]{article}

\usepackage{hyperref}
\usepackage{fullpage}
\usepackage{setspace}
\usepackage[UKenglish]{babel}
\usepackage{amssymb, amsfonts, amsmath, amsthm}
\usepackage{newpxtext,newpxmath}
\usepackage{mathtools}
\usepackage{enumitem}
\usepackage{bm}
\usepackage{algorithm}
\usepackage{algpseudocodex}
\usepackage{natbib}
\usepackage{xcolor}
\usepackage{url}
\usepackage{authblk}
\usepackage[capitalize]{cleveref}

\algrenewcommand\algorithmicrequire{\textbf{Input:}}
\algrenewcommand\algorithmicensure{\textbf{Output:}}

\theoremstyle{plain}
\newtheorem{theorem}{Theorem}
\newtheorem{proposition}{Proposition}
\newtheorem{lemma}{Lemma}
\newtheorem{corollary}{Corollary}

\newtheorem{observation}{Observation}
\theoremstyle{definition}
\newtheorem{example}{Example}

\newtheorem{assumption}{Assumption}

\crefname{assumption}{Assumption}{Assumptions}
\crefname{observation}{Observation}{Observations}
\crefname{claim}{Claim}{Claims}
\crefname{objective}{Objective}{Objectives}

\newcommand{\ab}{\bm{a}}

\newcommand{\xb}{\bm{x}}
\newcommand{\yb}{\bm{y}}
\newcommand{\zb}{\bm{z}}

\newcommand{\qb}{\bm{q}}
\newcommand{\ssb}{\bm{s}}
\newcommand{\eb}{\bm{e}}
\newcommand{\pb}{\bm{p}}
\newcommand{\rb}{\bm{r}}

\newcommand{\bids}{{\mathcal{B}}}
\newcommand{\bid}{{\bm{b}}}
\newcommand{\auction}{\mathcal{A}}
\newcommand{\market}{\mathcal{M}}

\usepackage{xcolor}

\newcommand{\R}{\mathbb{R}}

\DeclareMathOperator{\argmax}{argmax}
\DeclareMathOperator{\conv}{conv}

\begin{document}

\title{
Solving the Arctic Product-Mix Auction

\bigskip

\textsc{preliminary draft --- comments welcome}}

\author[1]{Elizabeth Baldwin\thanks{elizabeth.baldwin@economics.ox.ac.uk}}
\author[1]{Paul Klemperer\thanks{paul.klemperer@nuffield.ox.ac.uk}}
\author[2]{Edwin Lock\thanks{edwin.lock@kcl.ac.uk}}

\affil[1]{University of Oxford}
\affil[2]{King's College London}

\maketitle

\onehalfspacing

\begin{abstract}
Arctic product-mix auctions allow budget-constrained bidders to express preferences over multiple substitute goods. We give the  first strongly polynomial algorithm to find the competitive equilibria when the auctioneer has separable convex piecewise-linear costs. It combines a strongly polynomial reduction to the costless Arctic auction with an existing strongly polynomial algorithm for costless Arctic auctions. We also give a polynomial-time Turing reduction from costless Arctic auctions to linear Fisher markets, and a strongly polynomial many-to-one reduction to linear Arrow--Debreu exchange markets. If a good is sold in positive quantity in two competitive equilibria, its price is the same in both.


\end{abstract}

\section{Introduction}
\label{section:introduction}

The Arctic product-mix auction was developed for the Icelandic government in 2015 to allow the owners of the ``offshore'' accounts that it had blocked during the financial crisis to exchange their funds for alternative financial assets. Multiple substitute assets---cash, and bonds of various durations and terms, denominated in either kr\'onur or euros---would be offered. The government was therefore interested in using a version of the Bank of England's (BoE's) product-mix auction, an auction in which the auctioneer and the bidders are able to express rich preferences about how their trade-offs between substitute assets depend on the quantities that they, respectively, allocate or are allocated. An important distinction between the Icelandic government's context and the BoE's, however, was that the owner of any blocked account had a fixed quantity of funds available. That is, the Icelandic auction involved budget-constrained bidders.%
\footnote{The standard and Arctic product-mix auctions were developed by Klemperer (\textit{pro bono}) for the BoE and the Icelandic government, respectively.  The government hired dotEcon to program and test the Arctic design, and publicly announced in June 2015 that it would run an (Arctic) product-mix auction. See \citep{Klemperer2008,Klemperer2010,Klemperer} and \citet[Section 5.4]{BaldwinKlempererLock2024} for details of the BoE's product-mix auction, \citet[Appendix IB]{Klemperer} and~\citet[Section 5.5]{BaldwinKlempererLock2024} for more information about the Arctic auction, and \citet[Appendix IIA]{Klemperer} for more details of the Icelandic context. Lock's implementation of the auction is at \url{http://pma.nuff.ox.ac.uk/}.}

A budget-constrained bidder cares about price \textit{ratios}. If she has budget $\beta$, values asset $J$ at $r_J$ per unit, and faces price $p_J$, she would obtain utility $(r_J-p_J)(\beta/p_J)$ by exchanging her blocked funds for asset $J$. She therefore prefers to exchange her funds for asset $A$ rather than for asset $B$ if
\[
    (r_A-p_A)\frac{\beta}{p_A}
    > (r_B-p_B)\frac{\beta}{p_B},
    \qquad\text{or, equivalently,}\qquad
    \frac{r_A}{p_A}>\frac{r_B}{p_B}.
\]
By contrast, a bidder with quasi-linear preferences---that is, without a budget constraint, as better approximates the BoE's bidders---cares about price \textit{differences}: she prefers a unit of $A$ to a unit of $B$ if $r_A-p_A>r_B-p_B$.\footnote{In both the budget-constrained and quasi-linear preferences cases, the bidder prefers exchanging asset $A$ to no trade only if $r_A>p_A$.}

To be efficient, the bidding language for Iceland's auction therefore had to be modified from that used by the BoE. The auction was named ``Arctic'' both to reflect its planned use and because its bidders' preferences cannot be described ``tropically'', unlike those of the BoE's bidders.%
\footnote{Tropical-geometric methods apply to the quasi-linear preferences expressed by the standard product-mix auction language and yield, among other things, results on the existence of competitive equilibrium (see \citet{BaldwinKlemperer2014,BaldwinKlemperer2019}). Hard budget constraints take preferences outside that framework.} The Arctic auction was not, in the end, used in Iceland because political circumstances changed,%
\footnote{Detailed rules were to be published in April 2016, but a political crisis that month led to early elections and the plan to run the auction was subsequently abandoned. \citet[Appendix IIA]{Klemperer} provides a fuller history of the proposed implementation.} 
but it has subsequently been proposed by IMF staff for restructuring debts in other contexts.%
\footnote{\citet{willems2021auction} sets out an auction-based sovereign debt restructuring mechanism based on the arctic auction: creditors bid to exchange their claims for alternative debt instruments, while the debtor---and potentially outsiders offering new money---chooses supply schedules across those instruments. \citet[Appendix IIB]{Klemperer} discusses Mexico's ``Brady bond'' restructuring which let creditors choose among discount bonds, par bonds and new money, but pre-determined the terms and relative pricing of the three options; he explains that an arctic auction that set the relative prices and quantities in the light of the creditors' and the debtor's expressed preferences would have been more efficient. \par \citet[Appendix IIB]{Klemperer} also discusses other applications, such as using an arctic auction to offer the shareholders of a firm that is being acquired or restructured a choice of alternative securities in return for their shares—this would be more efficient than using a standard mechanism to compensate the firm’s shareholders.}

The algorithm originally developed and programmed for the Icelandic implementation, and subsequently analysed by \citet{Fichtl2022}, searches for candidate prices using the geometry of bidders' demands. It is practical when the number of goods is small---such as three or four---and the number of bids is not too large, but it would not be efficient in larger contexts.\footnote{The procedure finds the auctioneer's revenue-maximising envy-free outcomes at the same time as finding the competitive equilibria, but may need to examine $\mathcal B^n n!$ price vectors, where $n$ is the number of goods and $\mathcal B$ is the number of bids.} 

This paper's main contribution is to provide, to our knowledge, the first strongly polynomial algorithm for computing a competitive equilibrium of an Arctic auction with separable convex piecewise-linear costs---equivalently, separable stepwise-increasing marginal costs.%
\footnote{``Strongly polynomial'' means that the number of arithmetic operations is bounded by a polynomial in the numbers of goods, bids, and marginal-cost steps, independently of the number of digits used to represent the input.}
Our key step is a strongly polynomial reduction, the ``cost-removal reduction'', from the Arctic auction with costs to the Arctic auction without costs.%
\footnote{An Arctic auction without costs is also sometimes called a quasi-linear Fisher market.} 
Combined with \citet{GargTaherijamVazirani2026}'s strongly polynomial algorithm for the Arctic auction without costs, this yields our result. 

Separable piecewise-linear costs are a natural finite computational model: any continuous convex function can be approximated arbitrarily closely by a convex piecewise-linear function. They are also a natural way for an auctioneer to express how much of each asset it is willing to supply at each price.

\citet{GargLockVazirani2026} give the only other polynomial-time algorithm for solving the Arctic auction with (separable convex piecewise-linear) costs, but its running time depends on the bit length of the input, so it is \textit{not} strongly polynomial. Our reduction is also conceptually simpler, and we expect it to lead to more efficient implementations in practice.

We also provide a second path to a strongly polynomial algorithm, namely we follow our (strongly-polynomial) ``cost-removal reduction'' by first a strongly-polynomial ``exchange-market reduction'' from a costless Arctic auction to a linear Arrow--Debreu exchange market, and then the strongly polynomial exchange-market algorithm of \citet{GargVegh2023}.

Our third contribution, the ``Fisher-market reduction'' from Arctic auctions without costs to linear Fisher markets, allows alternative combinatorial, convex-programming, and price-adjustment methods to solve Arctic auctions with separable piecewise-linear costs, using existing Fisher market algorithms \citep{DPSV,Orlin2010,ChenJiangSo2025}.

Our reductions exploit the connection between Arctic auctions and standard market models. In particular, in an ordinary (linear) Fisher market, supplies are fixed and buyers spend fixed budgets on goods for which they have linear values. On the bidders' side, an Arctic auction is a quasi-linear Fisher market: a bidder may keep some of her budget when no purchase provides sufficient value. We can represent this choice in an ordinary Fisher market by adding a cash good. If both its value and its price are $1$, receiving one unit of cash is exactly equivalent to leaving one unit of the bidder's budget unspent. In this particular case, therefore, the Fisher-market-with-cash-good is equivalent to an Arctic auction without costs. 

The auction with costs, however, is \textit{not} equivalent to an ordinary Fisher market, because the quantity supplied depends on the prices. So a critical part of our contributions is the ``cost-removal reduction'' from the Arctic auction with costs to one without costs. 

The ``cost-removal reduction'' proceeds by dividing each good into separate \emph{varieties}, one for each step of its marginal-cost curve, and gives each variety a fixed supply equal to the length of that step. The original bidders regard all varieties of the same good as perfect substitutes. For each positive-cost variety, we add a synthetic ``buyback'' bid that wants only that variety, values it at its marginal cost, and has just enough budget to take the whole step at that price. A price below the marginal cost would make this bid demand more than the available supply and so cannot occur in equilibrium. At the marginal cost, the bid can absorb whatever regular bidders do not buy; above the marginal cost, it demands nothing and the regular bidders receive the entire step. The buyback bid therefore represents the auctioneer retaining units that are not worth selling.

These buyback bids  transform the auction with stepped marginal costs for each good into one with fixed supplies and no explicit costs: the auctioneer's supply decision has been expressed through the buyback bids. At equilibrium, lower-cost varieties are sold to regular bidders before higher-cost varieties, and the price of an original good is recovered as the lowest price among its varieties. (All varieties purchased by true (i.e., not synthetic) bids will be sold at this price.) Collapsing the varieties back into goods preserves both the bidders' demands and the auctioneer's competitive equilibrium allocation. \Cref{theorem:reduction-costs-costless} formalises this strongly polynomial ``cost-removal reduction'', and combining it with \citet{GargTaherijamVazirani2026}'s strongly polynomial algorithm for an Arctic auction without costs gives a strongly polynomial algorithm for the auction with costs. 

An alternative route to solving the auction with costs is to follow the ``cost-removal reduction'' by our ``Fisher-market reduction'' which maps an Arctic auction without costs to a linear Fisher market [(\cref{theorem:fisher-reduction})]. If at the Fisher equilibrium for the original goods without cash every bidder's most-preferred real good(s) give at least as much value per unit of budget as retaining her money, that equilibrium already solves the Arctic auction without costs. Otherwise, we add the cash good described above and binary-search for a supply at which its equilibrium price is exactly $1$. This gives a polynomial-time Turing reduction, although the search time may depend on the bit length of the input. Adding any existing Fisher market algorithm then completes an algorithm for an Arctic auction with costs, though such an algorithm is polynomial, not strongly polynomial.

A second way to obtain a \textit{strongly} polynomial algorithm for an Arctic auction with costs is by going via a standard linear Arrow--Debreu exchange market rather than via a Fisher market. In an exchange market, both buyers and sellers are endowed with real goods, and cannot buy more than they can afford from selling their endowment.  To transform the Arctic auction without costs into this setting, we need add only one good, money, and one agent, the seller.  Buyers are endowed with money according to their original budgets, while the seller is endowed with the supply of all other goods. The homogeneity of exchange-market prices allows the price of money to be normalised to $1$.  Money is either retained by buyers, or ``bought'' by the seller, so this avoids the search for the correct cash supply: its supply is simply the sum of original budgets.  Thus we obtain our ``exchange-market reduction'' (\cref{theorem:exchange-reduction}), which is strongly polynomial as we have added only one good and one agent. So combining the ``cost-removal reduction'' with our ``exchange-market reduction'' and  \citet{GargVegh2023}'s exchange-market algorithm yields a second strongly polynomial algorithm for an Arctic auction with costs.

Finally, we observe that since every Arctic auction without costs has an equilibrium (see, e.g., \citet{FinsterGoldbergLock2025}), our cost-removal reduction establishes the existence of equilibrium for Arctic auctions with (separable convex piecewise-linear) costs (\cref{corollary:existence-with-costs}). And we show the equilibrium prices in the Arctic auction with costs are essentially unique: any good that is sold in strictly positive quantity has the same price in all equilibria (\cref{proposition:price-uniqueness}). 

The paper proceeds as follows. Section~2 defines the Arctic auction and the Fisher market. \Cref{section:cost-reduction} gives the ``cost-removal reduction'' from the Arctic auction to an auction without costs (\cref{theorem:reduction-costs-costless}). \Cref{section:exchange-reduction} gives the reduction from an auction without costs to an exchange market (\cref{theorem:exchange-reduction}), and \cref{section:fisher-reduction} gives the reduction from an auction without costs to a Fisher market (\cref{theorem:fisher-reduction}). \Cref{section:uniqueness} studies the uniqueness of equilibrium prices.

\section{The Arctic Product-Mix Auction}

In the Arctic product-mix auction, or simply Arctic auction, bidders compete for goods $[n]\coloneqq\{1,\ldots,n\}$. A \textit{bundle} is a vector $\xb \in \R^n_{\geq 0}$ specifying the quantity of each good it contains.

Bidders each submit one or more bids. An \textit{arctic bid} $(\beta;\rb)$ consists of a budget $\beta \in\R_{\geq 0}$ and a root vector $\rb$ whose coordinates $r_i \in \R_{\geq 0}$ are the values the bid places on the goods $i \in [n]$. The bid can be interpreted as an agent with quasilinear utility defined by linear valuation $v(\xb) = \rb \cdot \xb$ and budget $\beta$. So at prices $\pb > \bm{0}$, the bidder has utility $\rb \cdot \xb - \pb \cdot \xb$ for all non-negative bundles such that $\pb \cdot \xb \leq \beta$ and utility $-\infty$ for all other bundles.

The bid's demand $D(\pb)$ consists of the utility-maximising bundles, so
\[
    D(\pb) \coloneqq \argmax_{\xb \in \R^n_{\geq 0},\ \pb \cdot \xb \leq \beta} (\rb - \pb) \cdot \xb;
\]
spending the bid's budget $\beta$ only on good $i$ obtains $\frac{\beta}{p_i}$ units of that good, so utility equal to $( \frac{r_i}{p_i}-1 ) \beta$. 

\begin{observation}
\label{observation:bang-per-buck}
When $\beta>0$, a positive amount of good $i$ can occur in a demanded bundle at prices $\pb > \bm{0}$ if and only if $r_i \geq p_i$ and $i$ maximises the ``bang per buck'', $\frac{r_i}{p_i}$, among all the goods $i \in [n]$.
\end{observation}

A bidder who prefers a portfolio of different goods can divide her budget across multiple separate bids; the demand of a set of arctic bids is simply the Minkowski sum of each bid's demand. Since each bid has its own budget constraint and is optimised independently, we henceforth treat every submitted bid as a separate bidder and denote the resulting bids by $\bid^1, \ldots, \bid^m$.

The auctioneer has separable costs for supplying the goods. For a bundle $\yb$, write
\[
C(\yb)=\sum_{i \in [n]} C_i(y_i),
\]
where each $C_i$ is a convex piecewise-linear cost with $K_i\geq 1$ marginal-cost steps, which increase with $y_i$. For good $i$, step $k \in [K_i]$ has width $s_{ik}>0$ and marginal cost $c_{ik}\geq 0$. Adjacent equal-cost steps are merged, so
\[
0 \leq c_{i1}, \qquad c_{ik}<c_{i,k+1}\quad(k=1,\ldots,K_i-1).
\]
Let $S_{i0}=0$ and $S_{ik}=\sum_{h=1}^k s_{ih}$. The final breakpoint $s_i\coloneqq S_{iK_i}$ is the supply cap for good $i$. If $y_i\in [S_{i,l-1},S_{il}]$, then
\[
C_i(y_i)
\coloneqq
\sum_{k<l} c_{ik}s_{ik}
+ c_{il}(y_i-S_{i,l-1}).
\]
For $y_i>s_i$, set $C_i(y_i)=\infty$.
The Arctic auction without costs is the special case in which $C_i(y_i)=0$ for $0\leq y_i\leq s_i$ and $C_i(y_i)=\infty$ for $y_i>s_i$. Thus the supply cap remains, but there is no cost for selling any feasible quantity.

Given prices $\pb$, the auctioneer's profit from selling bundle $\yb$ is
\[
\pb\cdot\yb-C(\yb)
= \sum_{i \in [n]} \left(p_i y_i-C_i(y_i)\right).
\]
An allocation $\xb^1,\ldots,\xb^m$ sells aggregate bundle $\yb=\sum_{j=1}^m\xb^j$. A competitive equilibrium consists of prices $\pb$ and allocations such that every bid receives a demanded bundle and
\[
\yb \in \argmax_{\zb\in \R^n_{\geq 0}} \left\{ \pb\cdot\zb-C(\zb) \right\}.
\]

Throughout, we make the following standing assumption on the bids of an Arctic auction.

\begin{assumption}
\label{assumption:budgets-values}
Without loss of generality, we assume that every bid has positive value for at least one good (otherwise we could remove the bid), and for each good $i \in [n]$ there is a bid $\bid^j$ with $\beta^j > 0$ and $r^j_i > 0$ (otherwise we could remove the good).
\end{assumption}

A competitive equilibrium always exists under \cref{assumption:budgets-values}; this is established in \citet{FinsterGoldbergLock2025} for the costless case, and our first reduction extends existence to the case with stepwise marginal costs (\cref{corollary:existence-with-costs}). Moreover, competitive equilibrium prices must be strictly positive. Finally, in the Arctic auction without costs, the auctioneer maximises profit by selling the full supply of every good.

\subsection{The Fisher Market}

A \textit{linear Fisher market} consists of goods $[n]$ with supplies $\ssb \in \R^n_{>0}$ and a set $1, \ldots, m$ of \textit{bidders}. Each bidder has a linear valuation $v(\xb) = \widetilde{\rb} \cdot \xb$ and a positive budget~$\widetilde{\beta}$, and has no value for any unused budget: at prices $\pb > \bm{0}$, the bidder's utility is $\widetilde{\rb} \cdot \xb$ for all bundles such that $\pb \cdot \xb \leq \widetilde{\beta}$, and $-\infty$ for all other bundles. We assume that every Fisher bidder has positive value for at least one good and that every good is valued positively by some bidder. We also denote a Fisher bidder by $(\widetilde{\beta};\widetilde{\rb})$. The demand of bidder $(\widetilde{\beta};\widetilde{\rb})$ at prices $\pb$ is
\[
\widetilde{I}(\pb) \coloneqq \argmax_{i \in [n]} \frac{\widetilde{r}_i}{p_i},
\qquad
\widetilde{D}(\pb) = \conv\left\{ \frac{\widetilde{\beta}}{p_i}\eb^i \mid i \in \widetilde{I}(\pb) \right\} \subset \R^n.
\]
Here $\eb^i$ denotes the $i$-th standard unit vector.
A competitive equilibrium of the Fisher market consists of prices $\pb \in \R^n_{>0}$ and allocations such that each bidder receives a bundle in its demand set and each good's supply is fully sold. Competitive equilibrium prices must be strictly positive.

\section{Reducing the Auction with Costs to the Auction without Costs}
\label{section:cost-reduction}

We show how an algorithm for solving the Arctic auction without costs can be used to solve the Arctic auction with costs. This is done via a reduction: a given auction with costs is turned into an auction without costs, and a solution to the latter is mapped back to a solution of the former.

The reduction removes the auctioneer's supply costs by turning each step of a good's marginal cost curve into a separate variety of that good. Regular bidders are made indifferent between the varieties belonging to the same original good, while additional buyback bids encode the auctioneer's willingness not to sell units whose price is below their marginal cost. In this way, the cost curve is represented entirely via ordinary bids and fixed supplies, so the resulting auction has no explicit seller costs.

Equilibria of the transformed auction have the following structure: lower-cost varieties are sold to regular bidders before higher-cost varieties, and the price of an original good can be recovered as the minimum price among its varieties. Once prices and allocations are collapsed back from varieties to goods, bidders still demand the same bundles and the auctioneer sells a profit-maximising quantity for each good at the given prices. Thus solving the transformed auction without costs provides a solution to the original auction with stepwise marginal costs.

\subsection{The Reduction}
Let $\auction$ be an Arctic auction with goods $1, \ldots, n$, stepwise marginal costs $(C_i)_{i \in [n]}$, and $m$ bidders with bids $\bid^1, \ldots, \bid^m$. We construct an auction $\widehat{\auction}$ without costs from the auction $\auction$ with costs.

In the new auction $\widehat{\auction}$, we refer to the objects to be sold as \textit{varieties}, to distinguish them from the goods in $\auction$. Specifically, for each good $i$ in the original auction, we create variety $(i,k)$ for each step $k \in [K_i]$ of $C_i$. So the total set of varieties in the new auction is $\{(i,k) \mid i \in [n], k \in [K_i] \}$. Each variety $(i,k)$ is available in supply $s_{ik}$ (the width of the $k$-th step of $C_i$).

$\widehat{\auction}$ has regular bids and buyback bids. For each bid $\bid^j = (\beta^j; \rb^j)$ in $\auction$, we create a \textit{regular bid} $\widehat{\bid}^j = (\widehat{\beta}^j; \widehat{\rb}^j)$ in $\widehat{\auction}$. This bid extends $\bid^j$ from goods to all varieties simply by making the bid indifferent between all varieties of the same good, and setting all these varieties' values to the bid's original value for the good. Formally, for each $j \in [m]$, we set the same budget, $\widehat{\beta}^j = \beta^j$, and define
\[
    \widehat{r}^j_{ik} \coloneqq r^j_{i} \quad \forall i \in [n], k \in [K_i].
\]

For each variety $(i,k)$ with $c_{ik}>0$, we also create a \textit{buyback bid} $\widehat{\bid}^{ik}$ that has value $c_{ik}$ for this variety and is not interested in any other varieties. Its budget $c_{ik}s_{ik}$ is exactly enough to purchase the entire supply $s_{ik}$ of variety $(i,k)$ at the marginal cost price $c_{ik}$, so
$
\widehat{\bid}^{ik} \coloneqq (c_{ik}s_{ik}; c_{ik} \eb^{ik}).
$
Here $\eb^{ik}$ denotes the standard unit vector corresponding to variety $(i,k)$.  Note that we do not need buyback bids for varieties with $c_{ik}=0$: equilibrium prices are strictly positive so the auctioneer always wants to sell the full supply $s_{ik}$.  (Also, a buyback bid for $c_{ik}=0$, following the pattern for $c_{ik}>0$, would have budget $0$ and value vector $\bm{0}$, which violates \cref{assumption:budgets-values}.)

The transformed auctioneer incurs no cost for supplying any variety $(i,k)$ up to its cap $s_{ik}$.

Finally, we check that $\widehat{\auction}$ is a legitimate instance.

\begin{observation}
\label{observation:well-formed}
The auction $\widehat{\auction}$ satisfies \cref{assumption:budgets-values}, so it is a valid Arctic auction without costs and has a competitive equilibrium.
\end{observation}
\begin{proof}
The auction $\auction$ satisfies \cref{assumption:budgets-values}. So, by construction of the reduction, every regular bid has positive value for at least one variety, and each variety has at least one regular bid with a positive budget and a positive value for the variety. Moreover, every buyback bid has positive value for its variety by construction. Hence $\widehat{\auction}$ satisfies \cref{assumption:budgets-values}. Existence of a competitive equilibrium now follows from \citet{FinsterGoldbergLock2025}.
\end{proof}

\subsection{Properties of \texorpdfstring{$\widehat{\auction}$}{Â}}

Throughout, we will let $(\widehat{\pb}, \widehat{\xb})$ denote a competitive equilibrium of $\widehat{\auction}$. Here $\widehat{\pb}$ are the prices for all the varieties, $\widehat{\xb}^j$ is the allocation to regular bid $\widehat{\bid}^j$, and, when $c_{ik}>0$, $\widehat{\xb}^{ik}$ is the allocation to buyback bid $\widehat{\bid}^{ik}$. We also denote the aggregate allocation to the regular bids by $\widehat{\zb} \coloneqq \sum_{j \in [m]} \widehat{\xb}^j$.

We establish initial properties about competitive equilibria of $\widehat{\auction}$.

\begin{observation}
\label{observation:buyback-implication}
For every variety $(i,k)$, we have $\widehat{p}_{ik} \geq c_{ik}$, and if $\widehat{p}_{ik}>c_{ik}$ then $\widehat{z}_{ik}=s_{ik}$.
\end{observation}
\begin{proof}
For $c_{ik}>0$, if $\widehat{p}_{ik}< c_{ik}$ then the buyback bid $\widehat{\bid}^{ik}$ demands more than the available supply. For $c_{ik}=0$, then $\widehat{p}_{ik} \geq c_{ik}$ follows from the strict positivity of equilibrium prices. Finally, if $\widehat{p}_{ik}>c_{ik}$ then no buyback bid demands variety $(i,k)$ (and none was created when $c_{ik}=0$), so market clearing implies $\widehat{z}_{ik}=s_{ik}$.
\end{proof}

\begin{lemma}
\label{lemma:min-price}
Every variety $(i,k)$ with $\widehat{z}_{ik} > 0$ satisfies $\widehat{p}_{ik} = \min_{k' \in [K_i]} \widehat{p}_{ik'}$.
\end{lemma}
\begin{proof}
Some regular bid demands $(i,k)$, and this bid has positive value for good $i$: a bid with zero value for good $i$ obtains negative utility from a positive amount of $(i,k)$ and utility zero from buying nothing. This bid values all varieties of good $i$ equally, so a strictly cheaper variety of good $i$ offers strictly higher bang per buck, contradicting the demand of $(i,k)$ (cf.~\cref{observation:bang-per-buck}). Hence $\widehat{p}_{ik}$ is minimal among the prices of the varieties of good $i$.
\end{proof}

\begin{lemma}
\label{lemma:monotonic-z}
Fix some good $i$. If $\widehat{z}_{ik} = 0$ for all $k \in [K_i]$, then $\widehat{p}_{ik} = c_{ik}$ for all $k \in [K_i]$. Otherwise, let $l \in [K_i]$ be the largest index for which $\widehat{z}_{il} > 0$. Then:
\begin{enumerate}[label=(\roman*)]
    \item \label{item:cutoff-structure} For every $k<l$, we have $\widehat{z}_{ik}=s_{ik}$ and $\widehat{p}_{ik}=\widehat{p}_{il}$. For every $k>l$, we have $\widehat{z}_{ik}=0$ and $\widehat{p}_{ik}=c_{ik}\geq\widehat{p}_{il}$.
    \item \label{item:marginal-price} We have $\widehat{p}_{il}\geq c_{il}$, with equality whenever $\widehat{z}_{il}<s_{il}$.
\end{enumerate}
\end{lemma}
\begin{proof}
If $\widehat{z}_{ik}=0$ for all $k \in [K_i]$, then $\widehat{z}_{ik}<s_{ik}$ for all $k$, so \cref{observation:buyback-implication} rules out $\widehat{p}_{ik}>c_{ik}$, implying $\widehat{p}_{ik}=c_{ik}$ for all $k \in [K_i]$.

So suppose there exists a largest index $l$ for which $\widehat{z}_{il}>0$. By \cref{observation:buyback-implication}, $\widehat{p}_{il}\geq c_{il}$. By \cref{lemma:min-price}, $\widehat{p}_{il} \leq \widehat{p}_{ik}$ for all $k \in [K_i]$, and $\widehat{p}_{ik}=\widehat{p}_{il}$ for every $k$ with $\widehat{z}_{ik}>0$.

Fix $k<l$. Strict monotonicity of the marginal costs yields $c_{ik}<c_{il}\leq\widehat{p}_{il}\leq\widehat{p}_{ik}$. Hence $\widehat{p}_{ik}>c_{ik}$, so \cref{observation:buyback-implication} implies $\widehat{z}_{ik}=s_{ik}>0$. We can therefore apply \cref{lemma:min-price} to obtain $\widehat{p}_{ik}=\widehat{p}_{il}$. Now fix $k>l$. Maximality of $l$ tells us that $\widehat{z}_{ik}=0<s_{ik}$, so \cref{observation:buyback-implication} rules out $\widehat{p}_{ik}>c_{ik}$ and leaves $\widehat{p}_{ik}=c_{ik}$. Together with $\widehat{p}_{il}\leq \widehat{p}_{ik}$, this establishes \ref{item:cutoff-structure}.

The lower bound in \ref{item:marginal-price} was established above. If $\widehat{z}_{il}<s_{il}$, then \cref{observation:buyback-implication} rules out $\widehat{p}_{il}>c_{il}$, so $\widehat{p}_{il}=c_{il}$.
\end{proof}

\subsection{Recovering a solution for the original auction}
\label{section:recovery}

Let $(\widehat{\pb}, \widehat{\xb})$ be a competitive equilibrium for the auction $\widehat{\auction}$. We map it as follows to a competitive equilibrium $(\pb, \xb)$ for the original auction $\auction$:
\begin{itemize}
    \item Set $p_i \coloneqq \min_{k \in [K_i]} \widehat{p}_{ik}$ for each $i \in [n]$.
    \item Set $x^j_i \coloneqq \sum_{k \in [K_i]} \widehat{x}^j_{ik}$ for each $i \in [n]$ and $j \in [m]$.
\end{itemize}

For notational convenience, we let $\zb \coloneqq \sum_{j \in [m]} \xb^j$ denote the aggregate allocation to the bids $\bid^1, \ldots, \bid^m$. The unsold bundle $\ssb - \zb$ is retained by the auctioneer.

It remains to prove that $(\pb, \xb)$ is indeed a competitive equilibrium for $\auction$. For this, we need to argue two properties. At prices $\pb$,
\begin{enumerate}[label=(P\arabic*)]
    \item \label{ce-bid-condition} each bid receives a utility-maximising bundle. This is equivalent to \cref{observation:bang-per-buck}, i.e., receiving a budget-feasible bundle containing only goods that maximise bang per buck and have bang per buck at least $1$ (this bundle can be empty if all goods have bang per buck at most $1$), and exhausting the bid's budget if the maximum bang per buck at $\pb$ is greater than $1$;
    \item \label{ce-auctioneer-condition} the auctioneer maximises their utility by selling $\zb$. Recall that the auctioneer's utility from selling bundle $\yb$ is $\pb \cdot \yb - \sum_{i \in [n]} C_i(y_i)$, so we need to show that $\zb$ maximises this. 
\end{enumerate}

\begin{lemma}
\label{lemma:envy-freeness}
The candidate competitive equilibrium $(\pb, \xb)$ obtained above satisfies \ref{ce-bid-condition}.
\end{lemma}
\begin{proof}
Fix a bid $\bid^j = (\beta^j; \rb^j)$ and its regular bid $\widehat{\bid}^j = (\widehat{\beta}^j; \widehat{\rb}^j)$ in $\widehat{\auction}$, and recall that $\widehat{\beta}^j = \beta^j$ and $\widehat{r}^j_{ik} = r^j_i$ for all $i \in [n]$ and $k \in [K_i]$.

Suppose first that $x^j_i = 0$ for all goods $i \in [n]$. Then $\widehat{x}^j_{ik} = 0$ for all varieties $(i,k)$, and so the regular bid $\widehat{\bid^j}$ in $\widehat{\auction}$ has bang per buck at most $1$ at prices $\widehat{p}$. By construction of the reduction, this implies that $\bid^j$ has bang per buck at most $1$ at $\pb$, so it satisfies \ref{ce-bid-condition} for the bid $\bid^j$ to receive the empty bundle at $\pb$.

Now suppose that $x^j_i > 0$ for at least one good $i$, and let $i$ denote a fixed such good. By the construction of $\xb^j$, we have $\widehat{x}^j_{ik} > 0$ for some variety $(i,k)$. Since $\widehat{\xb}^j$ is utility-maximising for $\widehat{\bid}^j$ at $\widehat{\pb}$, the variety $(i,k)$ maximises bang per buck for $\widehat{\bid}^j$ among all varieties and has bang per buck at least $1$ (\cref{observation:bang-per-buck}).

Now fix any good $i' \in [n]$, and choose $k' \in [K_{i'}]$ such that $p_{i'} = \widehat{p}_{i'k'}$; such a $k'$ exists by the construction of $\pb$ as a coordinatewise minimum. Since $\widehat{x}^j_{ik}>0$, we have $\widehat{z}_{ik}>0$, so \cref{lemma:min-price} implies that $(i,k)$ is a price-minimising variety of good $i$. Hence, by the construction of $\pb$, $p_i = \widehat{p}_{ik}$. Combining these two price identities with the value identities $r^j_i = \widehat{r}^j_{ik}$ and $r^j_{i'} = \widehat{r}^j_{i'k'}$, and using that $(i,k)$ maximises bang per buck for $\widehat{\bid}^j$ among all varieties for the middle inequality, we obtain
\[
\frac{r^j_i}{p_i} = \frac{\widehat{r}^j_{ik}}{\widehat{p}_{ik}} \geq \frac{\widehat{r}^j_{i'k'}}{\widehat{p}_{i'k'}} = \frac{r^j_{i'}}{p_{i'}}.
\]
The good $i'$ was arbitrary, so $i$ maximises bang per buck for $\bid^j$ at $\pb$. Moreover $(i,k)$ has bang per buck at least $1$, so the same two identities imply
\[
\frac{r^j_i}{p_i} = \frac{\widehat r^j_{ik}}{\widehat p_{ik}} \geq 1.
\]
Since the good $i$ with $x^j_i>0$ was also arbitrary, $\xb^j$ only contains goods that maximise bang per buck for $\bid^j$ at $\pb$, and every such good has bang per buck at least $1$.

It remains to show that $\bid^j$ is budget-feasible and exhausts its budget whenever the maximum bang per buck at $\pb$ exceeds $1$. We first record that the two bids spend the same amount in total. The argument of the previous paragraph applies to every variety $(h,k)$ with $\widehat{x}^j_{hk}>0$ and shows that $p_h = \widehat{p}_{hk}$ for each such variety. Hence, for every good $h \in [n]$,
\[
p_h x^j_h
= p_h \sum_{k \in [K_h]} \widehat{x}^j_{hk}
= \sum_{k \in [K_h]} \widehat{p}_{hk} \widehat{x}^j_{hk},
\]
since the terms with $\widehat{x}^j_{hk}=0$ contribute nothing to either side. Summing over $h$ yields $\pb \cdot \xb^j = \widehat{\pb} \cdot \widehat{\xb}^j$. Since $\widehat{\xb}^j$ is budget feasible and $\widehat\beta^j=\beta^j$, it follows that $\pb \cdot \xb^j\leq\beta^j$, so $\xb^j$ is budget feasible.

Now suppose some good $h$ has bang per buck greater than $1$ for $\bid^j$ at $\pb$. Choosing $k$ with $p_h = \widehat{p}_{hk}$, we have
\[
\frac{\widehat{r}^j_{hk}}{\widehat{p}_{hk}} = \frac{r^j_h}{p_h} > 1,
\]
so the maximum bang per buck of $\widehat{\bid}^j$ in $\widehat{\auction}$ also exceeds $1$. Since $\widehat{\xb}^j$ is utility-maximising for $\widehat{\bid}^j$, this forces $\widehat{\bid}^j$ to exhaust its budget, so $\widehat{\pb} \cdot \widehat{\xb}^j = \widehat{\beta}^j$. Combining with the previous paragraph and $\widehat{\beta}^j = \beta^j$, we get $\pb \cdot \xb^j = \beta^j$, and so $\bid^j$ exhausts its budget too. Hence, $\xb^j$ is utility-maximising for~$\bid^j$.
\end{proof}

\begin{lemma}
\label{lemma:auctioneer-condition}
The candidate competitive equilibrium $(\pb, \xb)$ obtained above satisfies \ref{ce-auctioneer-condition}.
\end{lemma}
\begin{proof}
Since both the payoff and the costs are separable across goods, it suffices to check each good separately. The auctioneer's payoff from good $i$ is $g_i(y_i) \coloneqq p_i y_i-C_i(y_i)$ for $y_i \in [0,s_i]$, which is continuous and piecewise linear with slope $p_i-c_{ik}$ on the interval $(S_{i,k-1},S_{ik})$ of step $k$. As the marginal costs $c_{ik}$ are strictly increasing in $k$, the slopes of $g_i$ are strictly decreasing in $k$, so $g_i$ is concave. Hence $z_i \in [0,s_i]$ maximises $g_i$ if the slope of every step lying below $z_i$ is non-negative, the slope of every step lying above $z_i$ is non-positive, and, in case $z_i$ lies in the interior of a step, the slope of that step is zero. Here we say that step $k$ lies below $z_i$ if $S_{ik} \leq z_i$, and above $z_i$ if $S_{i,k-1} \geq z_i$.

So fix good $i$. If $\widehat{z}_{ik} = 0$ for all $k \in [K_i]$, then $z_i = 0$ and, by \cref{lemma:monotonic-z}, $\widehat{p}_{ik} = c_{ik}$ for all $k$. By the construction of $\pb$ and the strict monotonicity of the marginal costs, $p_i = \min_{k \in [K_i]} c_{ik} = c_{i1}$. No step lies below $z_i = 0$, and every step $k$ satisfies $c_{ik} \geq c_{i1} = p_i$, so every slope is non-positive and selling any positive quantity cannot increase the auctioneer's payoff.

Now suppose that $\widehat{z}_{ik} > 0$ for some $k$, and let $l$ be the largest index with $\widehat{z}_{il} > 0$. By \cref{lemma:min-price}, the variety $(i,l)$ attains the minimum price among the varieties of good $i$, so $p_i = \widehat{p}_{il}$ by the construction of $\pb$. Furthermore,
\[
z_i = \sum_{k \in [K_i]} \widehat{z}_{ik} = \sum_{k=1}^l \widehat{z}_{ik},
\]
where the second equality uses $\widehat z_{ik}=0$ for $k>l$, by \cref{lemma:monotonic-z}\ref{item:cutoff-structure}. The same item also ensures that $\widehat z_{ik}=s_{ik}$ for $k<l$, while $0<\widehat z_{il}\leq s_{il}$, so
\[
\sum_{k=1}^{l-1} s_{ik} < z_i \leq \sum_{k=1}^{l} s_{ik}.
\]

We now distinguish between three cases.
\begin{description}
    \item[Case 1:] $\widehat{z}_{il} = s_{il}$ and $l < K_i$, so $z_i = \sum_{k=1}^{l} s_{ik}$. The lower bound $p_i=\widehat{p}_{il}\geq c_{il}$ follows from \cref{lemma:monotonic-z}\ref{item:marginal-price}, while \cref{lemma:monotonic-z}\ref{item:cutoff-structure} applied to $k=l+1$ implies $p_i=\widehat{p}_{il}\leq c_{i,l+1}$.
    \item[Case 2:] $\widehat{z}_{il} = s_{il}$ and $l = K_i$, so $z_i = s_i$. \Cref{lemma:monotonic-z}\ref{item:marginal-price} implies that $p_i = \widehat{p}_{il} \geq c_{il}$.
    \item[Case 3:] $\widehat{z}_{il} < s_{il}$, so $z_i < \sum_{k=1}^{l} s_{ik}$. \Cref{lemma:monotonic-z}\ref{item:marginal-price} implies that $p_i = \widehat{p}_{il} = c_{il}$.
\end{description}

We check the optimality criterion in each case. In Case~1, the steps lying below $z_i$ are those with $k \leq l$, and they satisfy $p_i-c_{ik}\geq c_{il}-c_{ik}\geq 0$ since $c_{ik} \leq c_{il}$; the steps lying above $z_i$ are those with $k>l$, and they satisfy $p_i-c_{ik}\leq c_{i,l+1}-c_{ik}\leq 0$ since $c_{ik} \geq c_{i,l+1}$. In Case~2, the steps lying below $z_i$ are those with $k \leq l = K_i$, and they satisfy $p_i \geq c_{il} \geq c_{ik}$; no step lies above $z_i = s_i$, and no larger quantity is feasible. In Case~3, the quantity $z_i$ lies in the interior of step $l$. The steps lying below $z_i$ are those with $k<l$, and they satisfy $p_i-c_{ik}=c_{il}-c_{ik}>0$; the steps lying above $z_i$ are those with $k>l$, and they satisfy $p_i - c_{ik} = c_{il}-c_{ik}<0$. The slope on step $l$ itself is $p_i-c_{il}=0$, so the payoff is constant across step $l$ and every quantity in it, including $z_i$, is optimal. In each case $z_i$ therefore maximises $g_i$, so neither reducing nor increasing the quantity of good $i$ sold at price $p_i$ increases the auctioneer's payoff.
\end{proof}

Recall that we consider instances of the Arctic auction in which the inputs---bids, supply, and supply costs curves---are rational values specified in binary.

\begin{theorem}
\label{theorem:reduction-costs-costless}
Arctic auctions with inputs encoded as rationals in binary and stepwise marginal costs reduce in strongly polynomial time to costless Arctic auctions.
\end{theorem}
\begin{proof}
Write $K=\sum_{i \in [n]} K_i$ for the total number of marginal-cost steps. The construction of $\widehat{\auction}$ creates one variety for each marginal-cost step, at most one buyback bid for each step, and one regular bid for each original bid. Thus an auction with $n$ goods and $m$ bids becomes a costless auction with $K$ varieties and at most $m+K$ bids. Its value vectors contain $O(mK+K^2)$ entries under dense encoding (in which every entry of the vectors is specified) and $O(mK+K)$ under sparse encoding (in which only non-zero entries, together with their indices, are specified), so its explicit description has polynomial size. By \cref{observation:well-formed}, $\widehat{\auction}$ is a valid costless Arctic auction and has a competitive equilibrium, so the reduction always produces a solvable instance.

The forward map writes this polynomial-size instance and performs at most one multiplication $c_{ik}s_{ik}$ per step. Given a competitive equilibrium $(\widehat{\pb},\widehat{\xb})$ of $\widehat{\auction}$, the prices and allocations for $\auction$ are recovered by
\[
p_i \coloneqq \min_{k \in [K_i]} \widehat p_{ik},
\qquad
x_i^j \coloneqq \sum_{k \in [K_i]} \widehat x^j_{ik},
\]
which uses $O(K)$ comparisons and $O(mK)$ additions. The number of arithmetic operations in both maps is thus bounded by a polynomial in $n$, $m$ and $K$ alone, independently of the bit lengths of the input numbers. The forward map produces only polynomial-bit numbers, while the reverse map only takes minima and sums of at most $K_i$ coordinates, so its intermediate and output numbers have encoding length polynomial in the size of $(\widehat{\pb}, \widehat{\xb})$. Hence both maps are strongly polynomial. \Cref{lemma:envy-freeness,lemma:auctioneer-condition} prove the bidder and auctioneer optimality conditions for the recovered outcome, so it is a competitive equilibrium of $\auction$.
\end{proof}

\begin{corollary}
\label{corollary:existence-with-costs}
Every Arctic auction with stepwise marginal costs satisfying \cref{assumption:budgets-values} has a competitive equilibrium.
\end{corollary}
\begin{proof}
By \cref{observation:well-formed}, the transformed auction $\widehat{\auction}$ is a valid costless Arctic auction, so it has a competitive equilibrium \citep{FinsterGoldbergLock2025}. By \cref{lemma:envy-freeness,lemma:auctioneer-condition}, the recovered outcome is a competitive equilibrium of $\auction$.
\end{proof}

\section{Reducing from Arctic to Linear Exchange Markets}
\label{section:exchange-reduction}

In the Arctic auction without costs, bidders have fixed budgets and the auctioneer wishes to sell the entirety of a fixed supply. A linear Arrow--Debreu exchange market lets budgets and supply arise from endowments instead. This permits a direct reduction: the seller is endowed with the real goods, the bidders are endowed with money, and money is an ordinary good.

A \textit{linear exchange market} on goods $[n]$ consists of agents $a$ with initial endowments $\bm{\omega}^a \in \R^{[n]}$ and a linear utility function $u^a$ mapping each possible bundle $\xb$ of goods to a value $u^a(\xb) = \rb^a \cdot \xb$ for some linear coefficients $\rb^a \in \R^{[n]}$. At prices $\pb$, agent $a$ has income $\pb\cdot\bm{\omega}^a$ and chooses a utility-maximising bundle subject to this budget. A competitive equilibrium consists of prices and demanded allocations such that every good clears. Exchange equilibrium prices are homogeneous: multiplying all prices by the same positive constant does not change demands.

\subsection{The Reduction}
Let $\auction$ be a costless Arctic auction with binary-rational data, goods $[n]$, supplies $\ssb \in \R^n_{>0}$, and bids $\bid^k = (\beta^k;\rb^k)$ for $k \in [m]$. We assume without loss of generality that every budget $\beta^k$ is positive; zero-budget bids demand nothing and can be discarded. Note that discarding them preserves \cref{assumption:budgets-values}, since the bid whose existence it requires for each good already has positive budget.

We construct a linear exchange market $\mathcal{E}$ on goods $[n]_0$, where good $0$ is money. For each bid $\bid^k$, we create an agent endowed with $\beta^k$ units of money and no real goods, with linear utility
\[
u^k(\xb)=x_0+\sum_{i \in [n]} r^k_i x_i .
\]
We also create one seller agent endowed with $s_i$ units of each real good $i$ and no money, whose utility is $u^S(\xb)=x_0$.

\subsection{Correctness}
We first show that the constructed exchange market has a competitive equilibrium, and then that any such equilibrium yields a competitive equilibrium of the original auction.

\begin{lemma}
\label{lemma:exchange-existence}
The exchange market $\mathcal{E}$ has a competitive equilibrium.
\end{lemma}
\begin{proof}
Let $(\pb, \xb)$ be a competitive equilibrium of $\auction$; one exists, and its prices are strictly positive \citep{FinsterGoldbergLock2025}. Extend the prices by $p_0 \coloneqq 1$. Allocate to the agent of bid $\bid^k$ the bundle $\xb^k$ together with its unspent budget $\beta^k - \pb \cdot \xb^k$ in money, and allocate to the seller $\sum_{i \in [n]} p_i s_i$ units of money and no real goods.

Each agent receives a utility-maximising bundle. By \cref{observation:bang-per-buck}, $\xb^k$ is supported on goods that maximise bang per buck among $[n]_0$, where money has bang per buck $1$. If money is not among the maximisers, then $\bid^k$ exhausts its budget, so the unspent remainder $\beta^k - \pb \cdot \xb^k$ is zero and the agent receives $\xb^k$ alone. If money is a maximiser, the remainder may be positive, but it is allocated in money, which maximises bang per buck by assumption. In both cases, the agent's bundle is supported on bang per buck maximising goods and spends its full budget $\beta^k$, so it is utility-maximising in $\mathcal{E}$. The seller values only money and all prices are positive, so spending its entire revenue $\sum_{i \in [n]} p_i s_i$ on money is optimal.

It remains to check market clearing. As the prices are strictly positive and $\auction$ has no costs, auctioneer optimality implies $\sum_{k \in [m]} \xb^k = \ssb$, so every real good clears. The agents hold $\sum_{k \in [m]} (\beta^k - \pb \cdot \xb^k) = \sum_{k \in [m]} \beta^k - \pb \cdot \ssb$ units of money and the seller holds $\pb \cdot \ssb$, so the money good clears as well.
\end{proof}

\begin{theorem}
\label{theorem:exchange-reduction}
Costless Arctic auctions with binary-rational data reduce in strongly polynomial time to linear exchange markets by a many-one reduction.
\end{theorem}
\begin{proof}
The construction creates $m+1$ agents on $n+1$ goods by copying input numbers and inserting zeros and ones, so it is strongly polynomial, and \cref{lemma:exchange-existence} guarantees that $\mathcal{E}$ has a competitive equilibrium. It remains to show that any competitive equilibrium of $\mathcal{E}$ yields a competitive equilibrium of $\auction$.

So consider a competitive equilibrium of $\mathcal{E}$. First, $p_0>0$, since every agent values money and would demand it in unbounded quantity at price zero. Each bidder therefore has positive income $\beta^kp_0$. For every real good $i$, \cref{assumption:budgets-values} ensures that there is a positive-budget bidder with $r^k_i>0$, who would demand good $i$ in unbounded quantity if $p_i=0$. Thus all prices are positive. Since exchange equilibrium prices are homogeneous, normalise them so that $p_0=1$. Bidder $k$ then has budget $\beta^k$ and maximises
\[
x_0+\sum_{i \in [n]} r^k_i x_i
\quad\text{subject to}\quad
x_0+\sum_{i \in [n]} p_i x_i \leq \beta^k .
\]
By \cref{observation:bang-per-buck}, bidder $k$ buys only goods that maximise bang per buck among the real goods and money. Money has bang per buck $1$, so a real good is bought only when $r^k_i \geq p_i$, and any unspent budget is retained as money. Deleting the money coordinate therefore yields exactly the Arctic demand correspondence $D^k(\pb)$, by \cref{observation:money-good-demand}. The seller spends all revenue on money and demands no real good, so market clearing allocates every real good to the bidders. Full sale is auctioneer-optimal because $\auction$ has no costs and prices are positive. The restricted prices and bidder allocations are therefore a competitive equilibrium of $\auction$.

Given polynomial-bit rational equilibrium data, if allocations are represented by money flows $f_{ki} = p_i x^k_i$, first recover quantities as $x^k_i=f_{ki}/p_i$; this is invariant under common rescaling. Recovery then normalises $n+1$ prices by $p_0$ and deletes the money coordinate and seller allocation. These operations use polynomially many divisions and preserve polynomial encoding length, so the reverse map is also strongly polynomial.
\end{proof}

\section{Reducing from Arctic to Fisher markets}
\label{section:fisher-reduction}

The ``Fisher-market'' reduction turns an Arctic auction without costs into an ordinary Fisher market by introducing a cash good. This good plays the role of the Arctic bidder's outside option: spending money on a real good is attractive only when that good offers at least as much value per dollar as keeping the money. When the value and price of the cash good are both $1$, Fisher demand over the real goods agrees exactly with Arctic demand.

We first solve the Arctic auction as a Fisher market with the original goods, supply, and bids, interpreted as Fisher bidders, without the cash good. If every bidder obtains bang per buck at least $1$, every bidder weakly prefers the real goods she receives to keeping her money (\cref{observation:bang-per-buck}). Hence the Fisher equilibrium is already an Arctic equilibrium, so we are done. Otherwise, we add a cash good to the Fisher market and vary its supply until its equilibrium price is $1$. We show that such a supply exists and can be found via binary search, followed if necessary by rational reconstruction, in polynomial time.

Each binary search step fixes a cash supply and computes a Fisher equilibrium. If the cash price is above $1$, the search moves to larger supplies; if it is below $1$, it moves to smaller supplies. If a supply with cash price $1$ has been found, the Fisher equilibrium yields an Arctic equilibrium by removing cash entries. If the search does not hit that price exactly, we can recover such a rational supply from a sufficiently short search interval.

We first define the markets and state the algorithm, then prove the demand correspondence, the existence of a suitable supply, and the bounds needed for exact recovery.

\subsection{The two Fisher markets}
Let $\auction$ be an Arctic auction without costs conducted on goods $[n]$, with supply $\ssb \in \R^n_{>0}$ and with a collection of bids $\bids = \{\bid^1, \ldots, \bid^m \}$. Recall that each Arctic bid $\bid^k = (\beta^k; \rb^k)$ consists of a budget $\beta^k$ and a root $\rb^k$. We assume that all supplies, budgets and values are binary-encoded rationals. A zero-budget bid demands only the zero bundle and can be discarded, so we also assume $\beta^k>0$ for every $k\in[m]$.

Let $\market$ be the Fisher market on the same goods $[n]$ with supply $\ssb$ and bids $\bids$, interpreted as Fisher bidders with the same value vector and budget. Fix an equilibrium $(\pb^*,\xb^*)$ of $\market$, and write $\xb^{*k}$ for bidder $k$'s allocation. Define
\[
{\alpha}_k\coloneqq\max_{i\in[n]}\frac{r_i^k}{p^*_i},
\qquad
\alpha_{\min}\coloneqq\min_{k\in[m]}\alpha_k.
\]
Thus $\alpha_k$ is bidder $k$'s maximum bang per buck in $\market$, and $\alpha_{\min}$ is the smallest of these maxima. Every $\alpha_k$ is positive by \cref{assumption:budgets-values}. Linear Fisher equilibrium prices are unique \citep{JainVazirani2010}, so these quantities do not depend on the choice of equilibrium. 

Secondly, let $\widetilde{\market}$ be the market defined on goods $[n]_0\coloneqq\{0,1,\ldots,n\}$, where good $0$ is an additional ``cash'' good. The Fisher bidders $\widetilde{\bids}$ are obtained by extending each Arctic bid's value vector with a value of $1$ for cash; the budget is unchanged. Formally, $\widetilde{\bids} = \{ \widetilde{\bid}^1, \ldots, \widetilde{\bid}^m \}$, where each bid $\widetilde{\bid}^k = (\beta^k; \widetilde{\rb}^k)$ has budget $\beta^k$ and value vector $\widetilde{\rb}^k$ defined by $\widetilde{r}^k_0 = 1$ and $\widetilde{r}^k_i = r^k_i$ for all $i \in [n]$. In $\widetilde{\market}$, the supply of goods $i \in [n]$ is set to the same supply $s_i$ as in the Arctic auction, while the supply $s_0$ of the cash good is a parameter.

For $s_0>0$, let $\widetilde{\pb}(s_0)$ denote the unique equilibrium price vector of $\widetilde{\market}$ at cash supply $s_0$. Write $\widetilde p_i(s_0)$ for the price of good $i\in[n]_0$; in particular, $\widetilde p_0(s_0)$ is the cash price. Existence and uniqueness of equilibrium prices in linear Fisher markets are well known, so the function is well-defined.

We use an exact Fisher equilibrium oracle: on a market with rational values, budgets and supplies, it returns rational equilibrium prices and allocations of polynomial encoding length in binary; such equilibria exist for rational linear Fisher markets \citep{DPSV}. In particular, its output permits exact comparison of prices with $1$. The size of an oracle call includes the binary encoding length of its rational input.

To make the search exact in polynomial time, we use the bounded rational solution from \cref{lemma:rational-supply} and the reconstruction procedure in \cref{lemma:rational-reconstruction}. \Cref{algorithm:reduction-2} states the procedure formally.

\begin{algorithm}
\begin{algorithmic}[1]
\Require An Arctic auction $\auction$ without costs and with binary-rational data.
\Ensure A competitive equilibrium for $\auction$.
\State Discard zero-budget bids and construct $\market$ on goods $[n]$ and its cash extension $\widetilde{\market}$.
\State Compute a Fisher equilibrium $(\pb^*,\xb^*)$ of $\market$.
\State Let $\alpha_k\gets\max_{i\in[n]} r_i^k/p_i^*$ for every $k\in[m]$, and $\alpha_{\min}\gets\min_{k\in[m]}\alpha_k$.
\If{$\alpha_{\min}\geq1$}
    \State \Return $(\pb^*,\xb^*)$.
\EndIf
\State Let $U \gets 2\sum_{k \in [m]}\beta^k$, $a \gets 0$, and $b \gets U$. \Comment{A positive solution lies in $(0,U)$}
\State Compute the bit-length bound $L$ from \cref{lemma:rational-supply}.
\State Let $N$ be the smallest non-negative integer such that $U/2^N<2^{-2L-1}$.
\For{$t=1,\ldots,N$}
    \State $c \gets (a+b)/2$.
    \State Solve $\widetilde{\market}$ at cash supply $c$, obtaining an equilibrium with cash price $\widetilde p_0(c)$.
    \If{$\widetilde p_0(c)=1$}
        \State \Return the restriction of this Fisher equilibrium to goods $[n]$.
    \ElsIf{$\widetilde p_0(c)>1$}
        \State $a \gets c$.
    \Else
        \State $b \gets c$.
    \EndIf
\EndFor
\State Use \cref{lemma:rational-reconstruction} to find the unique rational $s_0\in[a,b]$ whose reduced numerator and denominator are at most $2^L$.
\State Solve $\widetilde{\market}$ at cash supply $s_0$. \Comment{$\widetilde p_0(s_0)=1$}
\State \Return the restriction of this Fisher equilibrium to goods $[n]$.
\end{algorithmic}
\caption{Reducing an Arctic auction to Fisher markets}
\label{algorithm:reduction-2}
\end{algorithm}

\subsection{Correctness}
The connection between the two markets follows from an alternative characterisation of Arctic demand using a ``money good'' $0$, described in Section 5 of \citet{FinsterGoldbergLock2025}.

\begin{observation}
\label{observation:money-good-demand}
For a given Arctic bid $\bid^k = (\beta^k; \rb^k)$ and prices $\pb$, set $r^k_0=p_0=1$ and let
\[
I^k_{\auction}(\pb)\coloneqq\argmax_{i\in[n]_0}\frac{r_i^k}{p_i}.
\]
Then
\[
D^k(\pb)=\conv\left\{\frac{\beta^k}{p_i}\eb^i\ \middle|\ i\in I^k_{\auction}(\pb)\right\}\subset\R^n,
\]
where the vector in the term indexed by $i=0$ is understood to be $\bm 0$. Equivalently, if the extended Fisher bidder $\widetilde{\bid}^k$ faces cash price $1$, the projection of its Fisher demand onto the real goods equals $D^k(\pb)$.
\end{observation}

\begin{lemma}
\label{lemma:demand-correspondence}
If a Fisher equilibrium of $\widetilde{\market}$ has cash price $1$, then deleting the cash good yields an Arctic equilibrium. Moreover, an equilibrium $(\pb^*,\xb^*)$ of $\market$ is an Arctic equilibrium if and only if $\alpha_{\min} \geq 1$.
\end{lemma}
\begin{proof}
For the first claim, \cref{observation:money-good-demand} shows that the projection of each bidder's allocation onto the real goods is Arctic-demanded. Every real good clears, and full sale is auctioneer-optimal because the auction has no costs.

For the second claim, bidder $k$ spends its entire Fisher budget on real goods with bang per buck $\alpha_k$. If $\alpha_{\min}\geq1$, these goods also maximise bang per buck when the money good, whose bang per buck is $1$, is included. Hence every $\xb^{*k}$ is Arctic-demanded by \cref{observation:money-good-demand}; since the real goods clear, $(\pb^*,\xb^*)$ is an Arctic equilibrium. Conversely, if $\alpha_{\min}<1$, choose a bidder $k$ with $\alpha_k<1$. Because this bidder spends its full budget on goods with bang per buck $\alpha_k$, its quasilinear utility is less than $0$, so $\xb^{*k}$ is not Arctic-demanded.
\end{proof}

We next record how Fisher equilibrium prices respond when the supply of one good increases. This will establish the monotonicity and continuity needed for the cash-supply search. We allow the perturbed good's initial supply to be zero to allow for zero supply of cash.

\begin{proposition}
\label{proposition:continuous-prices}
Consider two Fisher markets on goods $G$ with the same bidders and supplies $\ssb$ and $\widehat{\ssb}=\ssb+\delta\eb^{i'}$, where $\delta>0$. All initial supplies are positive except that $s_{i'}$ may be zero. If $\pb,\widehat{\pb}>0$ are respective equilibrium price vectors, then
\[
0\leq p_i-\widehat p_i\leq\frac{p_{i'}\delta}{s_i}, \qquad i \in G,
\]
with the upper bound asserted only when $s_i>0$.
\end{proposition}
\begin{proof}
Write $\beta^k$ and $\rb^k$ for the common budgets and value vectors of bidders $k\in[m]$.
We first show that $\widehat{\pb}\leq\pb$. Suppose otherwise, and let
\[
\lambda\coloneqq\max_{h\in G}\frac{\widehat p_h}{p_h}>1,
\qquad
S\coloneqq\{h\in G\mid\widehat p_h=\lambda p_h\}.
\]
Let $\mathcal K$ be the bidders who spend a positive amount on $S$ at $\widehat{\pb}$. If bidder $k\in\mathcal K$ receives good $j\in S$, then for every $h\notin S$,
\[
\frac{r_j^k}{p_j}
=\lambda\frac{r_j^k}{\widehat p_j}
\geq\lambda\frac{r_h^k}{\widehat p_h}
\geq\frac{r_h^k}{p_h}.
\]
The last inequality is strict if $r_h^k>0$, since $\widehat p_h<\lambda p_h$. If $r_h^k=0$, the first expression is strictly positive because $r_j^k>0$. Thus, at $\pb$, every bidder in $\mathcal K$ spends its entire budget on $S$. The old revenue from $S$ is therefore at least $\sum_{k\in\mathcal K}\beta^k$, while the new revenue from $S$ is at most that amount. If the old revenue from $S$ is positive, then
\[
\sum_{h\in S}\widehat p_h\widehat s_h
=\lambda\sum_{h\in S}p_h\widehat s_h
\geq\lambda\sum_{h\in S}p_hs_h
>\sum_{h\in S}p_hs_h,
\]
a contradiction. If the old revenue is zero, then $S=\{i'\}$ and $s_{i'}=0$, since all other supplies and all prices are positive. The new revenue from $S$ is then $\widehat p_{i'}\delta>0$, again a contradiction. Hence $\widehat{\pb}\leq\pb$.

In both equilibria, total revenue equals the total budget. Consequently,
\begin{equation}
\label{eq:profit-difference}
\sum_{h\in G}(p_h-\widehat p_h)s_h=\widehat p_{i'}\delta.
\end{equation}
Every summand on the left is non-negative, so for any $i\in G$,
\[
0\leq(p_i-\widehat p_i)s_i
\leq\sum_{h\in G}(p_h-\widehat p_h)s_h
=\widehat p_{i'}\delta
\leq p_{i'}\delta.
\]
For every $i$ with $s_i>0$, division by $s_i$ proves the upper bound on the price decrease.
\end{proof}

\begin{lemma}
\label{lemma:cash-supply}
The cash-price function $\widetilde p_0$ is nonincreasing and continuous on $\R_{>0}$, and
\[
\lim_{s_0\downarrow0}\widetilde p_0(s_0)=\frac{1}{\alpha_{\min}},
\qquad
\lim_{s_0\downarrow0}\widetilde p_i(s_0)=p_i^*\quad(i\in[n]).
\]
For $U\coloneqq2\sum_{k\in[m]}\beta^k$, we have $\widetilde p_0(U)\leq\tfrac12$. Hence, if $\alpha_{\min}<1$, there exists $s_0\in(0,U)$ such that $\widetilde p_0(s_0)=1$.
\end{lemma}
\begin{proof}
We first establish monotonicity and continuity. For $0<a<b$, apply \cref{proposition:continuous-prices} with $i=i'=0$, initial cash supply $a$, and increase $\delta=b-a$. The old and new cash prices are $\widetilde p_0(a)$ and $\widetilde p_0(b)$, so
\[
0\leq\widetilde p_0(a)-\widetilde p_0(b)
\leq\frac{\widetilde p_0(a)(b-a)}{a}.
\]
Thus $\widetilde p_0$ is nonincreasing. On any compact interval of positive supplies, $\widetilde p_0(a)/a$ is bounded, so the same inequality gives local Lipschitz continuity.

It remains to identify the limiting prices as supply approaches zero. Set $\bar p_0\coloneqq1/\alpha_{\min}$. At prices $(\bar p_0,\pb^*)$, bidder $k$'s bang per buck for cash is $\alpha_{\min}\leq\alpha_k$. Hence adding a zero-supply cash good and allocating none of it leaves each allocation $\xb^{*k}$ demanded. These prices therefore support an auxiliary Fisher equilibrium in which only cash has zero supply.

Applying \cref{proposition:continuous-prices} to this auxiliary equilibrium, with cash supply increasing from zero to $s_0>0$, gives
\[
\widetilde p_0(s_0)\leq\frac{1}{\alpha_{\min}},
\qquad
0\leq p_i^*-\widetilde p_i(s_0)
\leq\frac{s_0}{\alpha_{\min}s_i}\quad(i\in[n]).
\]
Thus $\widetilde p_i(s_0)\to p_i^*$ as $s_0\downarrow0$.

To identify the limiting cash price, let $\beta_{\min}\coloneqq\min_k\beta^k>0$. Whenever $s_0<\beta_{\min}\alpha_{\min}$, total expenditure on cash satisfies
\[
\widetilde p_0(s_0)s_0\leq\frac{s_0}{\alpha_{\min}}<\beta_{\min}.
\]
No bidder can therefore spend its entire budget on cash, so every bidder consumes a real good.

For $s_0>0$, define the best bang per buck on real goods in the cash market by
\[
\widetilde{\alpha}_k(s_0)\coloneqq\max_{i\in[n]}\frac{r_i^k}{\widetilde p_i(s_0)},
\qquad
\widetilde{\alpha}_{\min}(s_0)\coloneqq\min_{k\in[m]}\widetilde{\alpha}_k(s_0).
\]
For the small supplies just considered, every bidder consumes a real good, so cash cannot offer strictly greater bang per buck. Hence $\widetilde{\alpha}_k(s_0)\geq1/\widetilde p_0(s_0)$ for every $k$. Some bidder consumes cash because its positive supply clears; this bidder also consumes a real good, so equality holds for that bidder. Therefore
\[
\widetilde{\alpha}_{\min}(s_0)=\frac{1}{\widetilde p_0(s_0)}.
\]
Since the real-good prices converge to $p_i^*$, we have $\widetilde{\alpha}_{\min}(s_0)\to\alpha_{\min}$. It follows that
\[
\lim_{s_0\downarrow0}\widetilde p_0(s_0)=\frac{1}{\alpha_{\min}}=\bar p_0.
\]

Revenue from cash cannot exceed the total budget, so
\[
\widetilde p_0(U)U\leq\sum_{k\in[m]}\beta^k=\frac U2.
\]
Thus $\widetilde p_0(U)\leq\tfrac12$. If $\alpha_{\min}<1$, the limit at zero is greater than $1$. Extending the cash-price function continuously to zero, the intermediate value theorem gives a supply $s_0\in(0,U)$ with cash price $1$.
\end{proof}

Bisection need not encounter a supply with cash price exactly $1$. The next two lemmas let us recover one exactly from a sufficiently short interval: first we bound the encoding length of some suitable supply, and then we reconstruct it.

\begin{lemma}
\label{lemma:rational-supply}
Suppose $\widetilde p_0(s_0)=1$ for some $s_0>0$. From the rational budgets, values and real-good supplies of $\widetilde{\market}$, one can compute in polynomial time an integer bound $L$, polynomial in their encoding length, so that some supply $u/v$, in lowest terms, satisfies $0<u,v\leq2^L$ and $\widetilde p_0(u/v)=1$.
\end{lemma}
\begin{proof}
Fix an arbitrary Fisher equilibrium $(\widetilde{\pb}^*,\widetilde{\xb}^*)$ for a positive cash supply $s^*_0$ with $\widetilde{p}^*_0=1$. For each bid $k$, let
\[
M^*_k \coloneqq \argmax_{i \in [n]_0} \frac{\widetilde r^k_i}{\widetilde p^*_i},
\qquad
a^*_k \coloneqq \max_{i \in [n]_0} \frac{\widetilde r^k_i}{\widetilde p^*_i}.
\]
We construct a rational polyhedron $P$ whose feasible points describe Fisher equilibria with cash price 1 in which each bid's allocation is supported on $M^*_k$. The variables are $q_i$ (which we interpret as the inverse of the price, $q_i = 1/p_i$) for each $i \in [n]_0$, bang-per-buck levels $a_k$ for bids $k \in [m]$, allocation quantities $x^k_i$ for bids $k \in [m]$ and goods $i \in [n]_0$, and the cash supply $s_0$. Add the constraints
\begin{subequations}
\renewcommand{\theequation}{C.\arabic{equation}}
\begin{alignat}{2}
q_0 &= 1, \label{c:cash-price}\\
q_i &\geq 0 &&\quad \forall i \in [n]_0, \label{c:q-nonneg}\\
\widetilde r^k_i q_i &= a_k &&\quad \forall k \in [m],\ i \in M^*_k, \label{c:bpb-tight}\\
\widetilde r^k_i q_i &\leq a_k &&\quad \forall k \in [m],\ i \in [n]_0 \setminus M^*_k, \label{c:bpb-slack}\\
x^k_i &\geq 0 &&\quad \forall k \in [m],\ i \in [n]_0, \label{c:x-nonneg}\\
x^k_i &= 0 &&\quad \forall k \in [m],\ i \in [n]_0 \setminus M^*_k, \label{c:support}\\
a_k \beta^k &= \sum_{i \in M^*_k} \widetilde r^k_i x^k_i &&\quad \forall k \in [m], \label{c:budget}\\
\sum_{k \in [m]} x^k_i &= s_i &&\quad \forall i \in [n], \label{c:clearing}\\
s_0 &= \sum_{k \in [m]} x^k_0. \label{c:cash-supply}
\end{alignat}
\end{subequations}

All coefficients of $P$ are rational input numbers. The point $(\widetilde{\qb}^*,\widetilde{\xb}^*,\ab^*,s^*_0)$, with $\widetilde q^*_i = 1/\widetilde p^*_i$, is feasible, so $P$ is non-empty and contains a point with $s_0>0$.

Conversely, let $(\qb,\xb,\ab,s_0)$ be any feasible point of $P$ with $s_0>0$. First note that $q_i>0$ for every good $i \in [n]_0$. Indeed, $q_0=1$. Moreover, the constraints involving the cash good imply $a_k \geq 1$ for every bid $k$: if $0 \in M^*_k$, then $a_k=\widetilde r^k_0q_0=1$, while if $0 \notin M^*_k$, then $1=\widetilde r^k_0q_0 \leq a_k$. For each real good $i \in [n]$, market clearing implies $\sum_k x^k_i=s_i>0$, so some bid $k$ has $x^k_i>0$. Then $i \in M^*_k$, and constraint~\labelcref{c:bpb-tight} implies $q_i>0$, since $a_k\geq 1$ and all values are non-negative. Define prices by $p_i = 1/q_i$. Constraint~\labelcref{c:cash-price} enforces $p_0=1$. Constraints~\labelcref{c:bpb-tight,c:bpb-slack} imply that $\max_{i \in [n]_0} \widetilde r^k_i q_i = a_k$, with equality attained on $M^*_k$. Goods outside $M^*_k$ may also attain this maximum, so bid $k$'s set of goods maximising bang per buck contains $M^*_k$; however, since $x^k_i = 0$ for $i \notin M^*_k$ by constraint~\labelcref{c:support}, the allocation is supported on maximum bang-per-buck goods, as required for equilibrium.
Moreover, since $\widetilde r^k_i q_i=a_k$ on the support of bid $k$'s allocation (and in particular $\widetilde r^k_0 = 1$ and $q_0 = 1$ for the cash good), we see that the bid exhausts its budget:
\[
\sum_{i \in [n]_0} p_i x^k_i
= \sum_{i \in M^*_k} \frac{x^k_i}{q_i} = \frac{1}{a_k}\sum_{i \in M^*_k} \widetilde r^k_i x^k_i
= \beta^k.
\]
Here the first equality uses $x^k_i=0$ for $i \notin M^*_k$, the second uses constraint~\labelcref{c:bpb-tight}, and the last uses constraint~\labelcref{c:budget}. Constraints~\labelcref{c:clearing,c:cash-supply} ensure that every real good is cleared and the cash supply is equal to the aggregate cash allocation. Hence the feasible point represents a Fisher equilibrium with cash price $1$ and cash supply $s_0$.

It remains to obtain a rational feasible point with $s_0>0$. The cash supply is bounded on $P$: if $0 \notin M^*_k$, then $x^k_0=0$, while if $0 \in M^*_k$, then $a_k=1$ and constraint~\labelcref{c:budget} implies $x^k_0 \leq \beta^k$. Therefore, by constraint~\labelcref{c:cash-supply}, $s_0 \leq \sum_{k \in [m]}\beta^k$. Since $P$ is a rational polyhedron and contains a feasible point with positive $s_0$, the rational linear program $\max\{s_0 \mid (\qb,\xb,\ab,s_0)\in P\}$ has a rational optimal solution with $s_0>0$. All variables of $P$ are non-negative, so $P$ contains no line and an optimal basic feasible solution exists. Standard polyhedral bounds \citep[Corollary 10.2a]{Sch-book} guarantee that there exists an optimal basic feasible solution that has polynomial bit length. More explicitly, we choose the bound using the maximum numbers of rows and columns and the maximum coefficient bit length over all possible demand sets; these quantities are polynomially bounded by the input, so the bound is computable without knowing the equilibrium demand graph. As shown above, any resulting feasible point with $s_0>0$ corresponds to a valid Fisher equilibrium with cash price~$1$.
\end{proof}

\begin{lemma}
\label{lemma:rational-reconstruction}
Let $0\leq a<b$ be rational and let $B\geq1$ be an integer with $b-a<1/(2B^2)$. In time polynomial in the encoding lengths of $a$, $b$ and $B$, one can find the unique rational in $[a,b]$, if any, whose reduced numerator lies in $[1,B]$ and whose denominator lies in $[1,B]$.
\end{lemma}
\begin{proof}
Two distinct rationals with denominators at most $B$ differ by at least $1/B^2$, so $[a,b]$ contains at most one candidate. The closed-interval continued-fraction algorithm \citep[Section 5.1]{Sch-book} returns a reduced rational of smallest denominator in $[a,b]$ in time polynomial in the encoding lengths of $a$ and $b$. If a candidate exists, this rational has denominator at most $B$ and hence equals the candidate by the separation bound. Checking its numerator and denominator therefore either returns the unique candidate or certifies that none exists.
\end{proof}

In the branch $\alpha_{\min}<1$, let $L$ be the bound from \cref{lemma:rational-supply} and set $B = 2^L$. We binary search on $[0,U]$, using $0$ only as a lower bound and never querying $\widetilde p_0$ there. The invariant is that every solution remains in the current interval. Since $\widetilde p_0$ is nonincreasing by \cref{lemma:cash-supply}, if the midpoint price is greater than $1$ every solution lies to the right; if it is less than $1$, every solution lies to the left; and if it is equal to $1$, we are done. We stop once the final interval has length less than $1/(2B^2)$.

If the search has not already found an exact solution, the final interval still contains the bounded-height solution guaranteed by \cref{lemma:rational-supply}. Hence \cref{lemma:rational-reconstruction} returns that solution. A final Fisher oracle call at this cash supply computes an equilibrium whose cash price is $1$, and its restriction is an Arctic equilibrium by \cref{lemma:demand-correspondence}.

\begin{theorem}
\label{theorem:fisher-reduction}
Costless Arctic auctions with binary-rational data reduce to linear Fisher markets by a polynomial-time Turing reduction. The reduction makes at most $N+2$ calls to an exact Fisher equilibrium oracle, where $N$ is polynomial in the input bit length.
\end{theorem}
\begin{proof}
\Cref{algorithm:reduction-2} first makes one oracle call for $\market$. If $\alpha_{\min}\geq1$, \cref{lemma:demand-correspondence} proves that this equilibrium can be returned immediately. Otherwise, \cref{lemma:cash-supply} guarantees a solution in $(0,U)$. The algorithm then makes at most $N$ binary-search calls and, unless one of them finds an exact solution, one call at the reconstructed supply. The worst-case number of oracle calls is therefore $N+2$, and correctness of the returned positive-supply equilibrium follows from \cref{lemma:demand-correspondence}.

The values $\alpha_k$ and their minimum $\alpha_{\min}$ are computed from the rational prices of $\market$ using $O(mn)$ exact rational operations and comparisons, and have polynomial encoding length. The minimal choice in the algorithm gives $N\leq2L+2+\max\{0,\lceil\log_2 U\rceil\}$, so $N$ is polynomial. Every queried supply is a rational multiple of $U$ with denominator at most $2^N$, and hence has polynomial encoding length. The oracle outputs permit exact comparison with $1$, so the entire procedure is a polynomial-time Turing reduction. Discarded zero-budget bids are assigned the zero bundle.
\end{proof}

\section{Uniqueness of Equilibrium Prices}
\label{section:uniqueness}

In Arctic auctions without costs, equilibrium prices are unique \citep{FinsterGoldbergLock2025}. With costs, if there exist two equilibria, then every good sold in positive quantity in both has the same price in the two equilibria, as we show in \cref{proposition:price-uniqueness}; non-uniqueness can therefore only involve goods unsold in at least one of them. \Cref{example:non-unique} illustrates this.

\begin{proposition}
\label{proposition:price-uniqueness}
Fix an Arctic auction with stepwise marginal costs. If a good is sold in positive quantity in two competitive equilibria, then its price is the same in both equilibria. In particular, any two competitive equilibria in which every good is sold in positive quantity have the same price vector.
\end{proposition}
\begin{proof}
Consider the transformation from \cref{section:cost-reduction} that replaces each marginal-cost step by a variety and adds the corresponding buyback bid. We first note that any competitive equilibrium $(\pb,\xb)$ of the auction with costs can be lifted to a competitive equilibrium of the transformed auction without costs in such a way that, for every good $i$ sold in positive quantity, the original price $p_i$ is the minimum price among the varieties of good $i$.

So suppose that $(\pb,\xb)$ is a competitive equilibrium of the auction with costs.  For each good $i$, refine the regular bidders' allocations of good $i$ into allocations of its varieties, preserving each bidder's total allocation of good $i$ and respecting each variety's supply, so that the aggregate quantity sold fills lower-cost varieties first. Auctioneer optimality in the original equilibrium already determines which varieties are filled this way; the refinement only chooses how the quantity of each filled variety is divided among the bidders, which is always possible since these quantities sum to the aggregate allocation of good $i$. If good $i$ is sold in positive quantity, let $l$ be the last variety of good $i$ that now receives a positive allocation from the regular bids. Set the prices $\widehat{p}_{ik}$ of each of these varieties.  If $k\leq l$ set $\widehat{p}_{ik}:=p_i$, and if $k>l$ set $\widehat{p}_{ik}:=c_{ik}$. The auctioneer's optimality conditions in the original equilibrium imply exactly the inequalities needed for these prices: earlier sold varieties have cost at most $p_i$, later unsold varieties have cost at least $p_i$, and if the last sold variety is only partly sold then $p_i=c_{il}$.

If good $i$ is not sold, set every variety price equal to its cost, $\widehat p_{ik}:=c_{ik}$. Original auctioneer optimality at zero sales implies $p_i\leq c_{i1}\leq c_{ik}$ for every $k$, so every variety of this unsold good is weakly more expensive than the original good. Since no regular bidder receives good $i$ in the original equilibrium, replacing good $i$ by these varieties cannot make any allocated regular bundle cease to be demanded.

Every zero-cost variety is fully allocated to regular bidders, since equilibrium prices are strictly positive and selling the entire zero-cost step is strictly profit-maximising. Thus every variety with unsold units has a corresponding buyback bid; allocate its unsold units to that bid. Regular bidders face the same bang per buck on every variety $k\leq l$ of a sold good as on the original good, while later varieties have price $\widehat{p}_{ik}=c_{ik}\geq p_i$; varieties of an unsold good have price at least its original price. Hence no new variety offers a better bang per buck, and the refined regular allocations remain demanded. The buyback bids demand their assigned amounts: if a variety price is above its cost then its buyback bid demands none of it, and if the price equals its cost then any amount up to its full supply is demanded. Thus the lifted outcome is a competitive equilibrium of the transformed costless auction, and for every good $i$ sold in positive quantity, $p_i$ is the minimum variety price for good $i$.

By \cref{observation:well-formed}, the transformed auction is a valid costless Arctic auction, so it has a unique equilibrium price vector \citep{FinsterGoldbergLock2025}. Therefore the minimum price among the varieties of any fixed good $i$ is unique. Since this minimum equals the original price in every equilibrium in which good $i$ is sold positively, the price of such a good is the same for each equilibrium.
\end{proof}

\begin{example}
\label{example:non-unique}
Consider two goods, each in supply $1$, and one bid with values $(1,1)$ and budget $\beta=1$. Good $1$ has marginal cost $0$ and good $2$ has marginal cost $2$, both up to their unit supply.

For every $p_2 \in [1,2]$, set $\pb=(1,p_2)$ and allocate one unit of good $1$ to the bidder. The bidder demands this bundle because good $1$ and money both have bang per buck $1$, while good $2$ has bang per buck $1/p_2 \leq 1$. The auctioneer also optimises: selling good $1$ is profitable, and the payoff from selling any amount $y$ of good $2$ is $(p_2-2)y \leq 0$. Thus every $\pb=(1,p_2)$ with $p_2 \in [1,2]$ is an equilibrium price vector.
\end{example}

\bibliographystyle{plainnat}
\bibliography{refs}

\end{document}